\documentclass[notitlepage,twocolumn,superscriptaddress,10pt,nofootinbib,
prx,longbibliography,floatfix
]{revtex4-2}

\usepackage{graphicx}% Include figure files
\usepackage{bm}% bold math
\usepackage{xcolor}
\usepackage{hyperref}
\usepackage{amsmath}
\usepackage{amssymb}
\usepackage{amsthm}
\usepackage{cleveref}
\usepackage[justification=RaggedRight]{caption}
\usepackage{subcaption}
\usepackage{xspace}
\usepackage{complexity}
\usepackage{bbold}
\usepackage{enumitem}
\usepackage{physics}
\usepackage{algorithm}
\usepackage{algorithmicx}
\usepackage{placeins}
\usepackage{balance}
\usepackage{float}

\definecolor{orange(colorwheel)}{rgb}{1.0, 0.5, 0.0}

\newcommand{\adj}{^{\dag}}
\newcommand{\ssec}[1]{\subsection{#1}}
\newcommand{\GHZ}{\mathrm{GHZ}}

\newtheorem{theorem}{Theorem}
\newtheorem{lemma}[theorem]{Lemma}
\newtheorem{corollary}[theorem]{Corollary}

\newtheorem*{rep@theorem}{\rep@title}
\newcommand{\newreptheorem}[2]{%
\newenvironment{rep#1}[1]{%
 \def\rep@title{#2 \ref{##1}}%
 \begin{rep@theorem}}%
 {\end{rep@theorem}}}
\makeatother

\newreptheorem{thm}{Theorem}
\newreptheorem{lem}{Lemma}

\newcommand{\bo}{b^{(1)}}
\newcommand{\bt}{b^{(2)}}

\begin{document}

% \preprint{APS/123-QED}

\title{Satisfiability Phase Transtion for Random Quantum 3XOR Games}
\author{Adam~Bene~Watts}
\email{adam.benewatts1@uwaterloo.ca}
\affiliation{Institute for Quantum Computing, University of Waterloo, Canada}
\author{J. William Helton}
\email{helton@math.ucsd.edu}
\affiliation{Mathematics Dept., UC San Diego}
\author{Zehong Zhao}
\email{zez045@ucsd.edu}
\affiliation{Mathematics Dept., UC San Diego}
\date{\today}

\newcommand{\questionSet}{\mathcal{Q}}
\newcommand{\questionVec}[1]{\vec{q}_{#1}}
\newcommand{\question}[2]{q_{#1}(#2)}
\newcommand{\players}{k}
\newcommand{\winParity}[1]{s_{#1}}
\newcommand{\response}[2]{x_{#1}^{(#2)}}
\newcommand{\clauseNumber}{m}
\newcommand{\clauseSet}{\mathcal{C}}
\newcommand{\clause}[1]{c_{#1}}
\newcommand{\game}{\mathcal{G}}
\newcommand{\qperf}{Q-perfect\xspace}
\newcommand{\Qperf}{Q-Perfect\xspace}
\newcommand{\cperf}{C-perfect\xspace}
\newcommand{\Cperf}{C-Perfect\xspace}
\newcommand{\qcperf}{Q-perfect and C-perfect\xspace}
\newcommand{\QCperf}{Q and C-Perfect\xspace}

%%%%%%%%%%%%%%%
\def\perfeq{Defining Equations}

\def\bbZ{{\mathbb Z}}

\def\ben{\begin{enumerate}}
\def\een{\end{enumerate}}

%%%%%%%%%%%%%%%%%%%%%%%%%%%%%%%%%

\date{\today}

\begin{abstract}

Recent results showed it was possible to determine if a modest size  3XOR game has a perfect quantum strategy.
We build on these
and give an explicit polynomial time algorithm
which constructs such a perfect strategy or refutes its existence.
This new tool lets us
 numerically study 
 the behavior of randomly generated 3XOR games with large numbers of questions. 
 
 A key issue is: how common are pseudotelephathy games (games with perfect quantum strategies but no perfect classical strategies)?
 Our experiments strongly indicate
that the probability of a randomly generated game being pseudotelpathic stays far from 1, indeed it is bounded below~0.15. 
 
%  A key issue is: how much more common are Q-perfect than C-perfect strategies? 
%  Our experiments strongly indicate
% that the probability of a game being pseudotelpathic (i.e. having a perfect quantum strategy but no perfect classical strategy)
% stays far from 1, indeed it is bounded below 0.15. 

We also
 find strong evidence that randomly generated 3XOR games undergo 
 both a quantum and classical %satisfiability
 ``phase transition", transitioning from almost certainly perfect to almost certainly imperfect as the ratio of number of clauses ($m$) to number of questions ($n$) increases. The locations of these two phase transitions appear to coincide at 
  $m/n \approx 2.74$.

% PRL abstracts have 600 characters maximum. paper $< 3750$ Words. 

\end{abstract}
\maketitle

\section{Introduction}
\label{sec:intro}

Nonlocal games describe experiments in which a \emph{verifier} tests two or more \emph{players}' ability to produce correlated responses without communicating.
In a round of a nonlocal game the verifier sends questions to each player and the players return responses without knowing the questions sent to other players.
The questions asked and responses received are then collectively scored by the verifier according to a scoring function known by both the players and the verifier before the game begins. By convention, the score the players receives takes value in the interval $[0,1]$.

The key objects of study for a nonlocal game are its values, each defined to be the optimal (supremum) score players sharing certain resources can achieve in expectation while playing the game. The supremum score the players can achieve using classical resources (i.e. shared randomness) is called the \emph{classical value} of the game. The supremum score non-communicating players can achieve if they are allowed to make measurements on shared finite dimensional entangled states is called the \emph{quantum value} of the game.\footnote{A third value, called the commuting operator value, denotes the supremum score players can achieve if they can make commuting measurements on a shared, possibly infinite dimensional, entangled state. This value can be larger than the quantum value, but two values will coincide for the games characterized in this paper, so we will mostly avoid discussing it.} 

Games whose quantum value exceeds their classical value are important objects in quantum information. 
The first known game with this property, called the CHSH game, is closely related to the test proposed in John Bell's famous 1964 paper~\cite{bell1964einstein}.\footnote{Though the nonlocal games formalism wasn't developed until later by Clauser, Horne, Shimony and Holt after whom the CHSH game is named.} 
Since then, many other games with superclassical quantum value have been discovered, and these games have found important applications in quantum crypography, quantum foundations, delegated quantum computing, and fundamental mathematics~\cite{brassard2005quantum,DQC:reichardt2013classical,ji2020mip,rauch2018cosmic,
big2018challenging,
QKD2:ekert1991quantum}. 
Yet little is known about the general conditions under which games can have superclassical values. How 
common are games with superclassical quantum value? Is this phenomenon rare? Or can it be viewed as a generic feature of nonlocal games?

A major obstacle to answering these questions comes from the fact that the quantum value of a nonlocal game is, in general, uncomputable -- even approximately.\footnote{Formally: in general it is undecidable whether the value of a game is less than 1/2 or greater than~1~\cite{ji2020mip}.} 
There are some known techniques for bounding this value: a decreasing series of upper bounds on this value is provided by the NPA/ncSoS algorithm, and a lower bounds can be obtained by exhaustive search over \emph{strategies} players can use for responding to questions. However the run-time of both these algorithms quickly become infeasible as the game size (measured in terms of the number of questions and responses) grows large. 
Moreover, the NPA/ncSoS upper bound on a game's value does not necessarily converge to the quantum value (as implied by the aforementioned uncomputability result). 

\subsection{Results in this paper}

This paper studies a well known family of nonlocal games called XOR games. 
Every XOR is based on a system of equations over $\mathbb{Z}_2$.
In a round of the game, players are sent questions representing single variables in one equation in the system, and win if their responses (interpreted as variable assignments) satisfy that equation. 
We classify XOR games in terms of their numbers of \emph{questions} and \emph{clauses}: an $n$ question and $m$ clause XOR game has $n$ different variable assignments that can be requested from each player and $m$ equations involved the system of equations defining the game.
We are particularly interested in XOR games whose classical or quantum value equals the maximum possible score of 1. We call those games \emph{\qperf} or \emph{\cperf}, respectively. 
A recent series of works \cite{watts2018algorithms,watts20203xor} showed that it is possible to decide if 3 player XOR (3XOR) game is \qperf in time $\poly(n,m)$. 

In the first half of this paper we build on techniques described in \cite{watts20203xor}.
% , and give an explicit algorithm for identifying \Qperf games.
We give a polynomial time algorithm which both identifies \Qperf games and describes a measurement strategy players can use to win them.
 It is highly efficient in practice, letting us easily study games with up to 100 questions per player and 300 clauses.  By contrast, the third level of the NPA hierarchy becomes infeasible when there are more then about 5 questions per player, and a naive implementation of the techniques proposed in \cite{watts20203xor} (implemented and discussed in \Cref{app:dual_equations}) becomes infeasible past about 30 questions per player.

This strategies described by our algorithm belong to a class of strategies called MERP strategies in \cite{watts2018algorithms}: in these strategies each player holds a single qubit of a shared GHZ state, and determines their responses to a question $j$ by measuring an observable of the form
\begin{align}
    \exp(i \sigma_z z^{(\alpha)}_j)  \sigma_x \exp(- i \sigma_z z^{(\alpha)}_j)
\end{align}
where $\sigma_x, \sigma_z$ are the Pauli X and Z operators and the rotation angle $z_j$ depends on the question $j$ asked. 
\cite{watts20203xor} showed such a strategy can always achieve value $1$ on 3XOR games with perfect quantum value, but an explicit algorithm for finding these strategies was not given until now.

In the second half of this paper we use the algorithm developed in the first half to study numerically the behavior of randomly generated 3XOR games with large numbers of questions and clauses. In particular, we study the probability that a randomly generated game with $n$ questions and $m$ clauses is a \emph{pseudotelepathy game}, i.e. a game which is \qperf but not \cperf. We call this probability, the \emph{pseudotelepathy probability}. At each fixed question number $n > 3$, we find this probability is maximized when the ratio of clauses to questions in the 3XOR game is near a constant $m/n \approx 2.74$. Somewhat mysteriously, we also find this maximized pseudotelepathy probability initially rises steadily with $n$ to peak at $n \approx 38$ and then falls off. 
The pseudotelepathy probability at the peak is about $0.14$.
 
We also observe strong evidence that the probability a 3XOR game is \qperf undergoes a sharp phase transition, meaning that there exists a constant $C_q$ such that (in the limit of large question number $n$) a randomly generated 3XOR game with $m$ clauses and $n$ questions is almost-certainly satisfiable provided $m/n < C_q$ and almost certainly satisfiable provided $m/n > C_q$. Based on the observed pseudotelepathy behavior of randomly generated 3XOR games and additional  numerical evidence, we conjecture the location $C_q$ of this phase transition coincides with the location $C_c$ of a sharp \cperf phase transition. 

\subsection{Readers Guide}

\Cref{sec:XOR_games_and_algorithm} reviews the definition of XOR games and some relevant results from~\cite{watts2018algorithms,watts20203xor}, then presents the new algorithm developed in this paper for deciding if 3XOR games are \qperf, as well as finding descriptions of the perfect quantum strategies when they exist. \Cref{sec:Numerics_and_phase_transitions} begins with a discussion of phase transitions, then presents our main numerical results. \Cref{sec:Discussion} summarizes the main results of the paper and presents some open questions. The appendices contain details deemed too technical for the main paper.

\def\cG{{\mathcal G}}
\def\cL{{\mathcal L}}

\section{The linear equations corresponding to perfect games}

\label{sec:XOR_games_and_algorithm}

In this section we overview of the mathematical and computational techniques used in this letter to identify \qperf XOR games. We begin by reviewing the definition of XOR games and some relevant results from the literature. 

    \subsection{XOR Games and Defining Equations}
    
    In this section we develop notation to describe 3 player XOR (3XOR) games, 
    which will be the main focus of this paper. 
    % While we don't discuss them in this paper, we note that $k$ player XOR games can be described similarly. 
    Following convention, we will call the 3 players Alice, Bob, and Charlie. 
    In a round of a 3XOR game $\game$ the verifier sends each player an integer question, which we label $a_i, b_i$ and $c_i$, respectively. Each question is drawn from a set $\{1,2,..., n\}$ of possible integer questions.
    Each players responds with a single bit answer, selected according to some strategy, which may involve randomness (in the classical case) or measurements of some quantum state (in the quantum case). 
    The players win the game if the sum of their answers mod 2 (i.e. the XOR of their answers) equals some winning \emph{parity bit} $s_i$.  
    We call each tuple $(a_i, b_i, c_i, s_i)$ consisting of questions and associated parity bit which together describe a possible round of the game $\game$ a \emph{clause} of game. 
    
    In all the XOR games we consider in this paper, we assume the verifier chooses the questions sent to the players and associated winning parity bit uniformly at random from a set of $m$ possible clauses.\footnote{Changing the probability of certain clauses being selected does not change whether or not a game is \qperf or \cperf (as long as all probabilities remain nonzero).
    Since that is what is studied in this paper, we can make this assumption freely.} 
    Thus, a game $\game$ can be completely described by specifying the set of clauses associated with the game. 
    
    We can also describe each clause $(a_i, b_i, c_i, s_i)$ of an 3XOR game $\game$ via an equation
  \begin{align}
     \label{eq:XOR_games_equation}
         x^{(1)}_{a_i} + x^{(2)}_{b_i} + x^{(3)}_{c_i} = s_i  \pmod{2}. 
     \end{align}
     Then the set of all clauses associated with $\game$ can be described via a system of equations, which we write in matrix form as 
     \begin{align}
     \label{eq:matrix_SOE}
            \Gamma x := 
            \begin{pmatrix}
             A & B & C
            \end{pmatrix}
            x
            = 
            S \pmod{2}.
     \end{align}
      Here $A, B, C$ are $n \times m$ matrices whose $i$-th row has a single nonzero entry in the position corresponding to $a_i, b_i, c_i$ respectively, $x$ is a $1 \times 3n$ vector of random variables, and $S$ is a $1 \times m$ vector listing the variables $s_i$. This system of equations given in \Cref{eq:matrix_SOE} specifies all the clauses of the associated 3XOR game and hence uniquely determines the game. For that reason, these equations are referred to as the \emph{defining equations} of the associated 3XOR game. 
      Note that we have not yet specified the field the variables $x$ can range over. This is deliberate; in the next section we will see how studying the defining equations of a XOR game over $\mathbb{Z}_2$ and $\mathbb{R}$ gives information about the classical and quantum values of a game, respectively. 
    
    \subsection{Computing the Value of XOR Games}
    
    We can study the quantum and classical values of an XOR game by considering possible \emph{strategies} players may use for mapping questions to responses. 
    A straightforward convexity argument shows that the classical value of a game can always be achieved by a \emph{deterministic strategy}, where players give fixed responses to each possible question. 
    Letting variables $x_j^{(\alpha)} \in \{0,1\}$ represent the response of player $\alpha$ to question $j$, we see a deterministic strategy described by these variables wins the round associated with clause $(a_i, b_i, c_i, s_i)$ iff these variables satisfy the associated (defining) equation
    \begin{align}
        x^{(1)}_{a_i} + x^{(2)}_{b_i} + x^{(3)}_{c_i} = s_i \pmod{2}.
    \end{align}
    It follows that the classical value of a game $\game$ is equal to the maximum fraction of satisfiable equations in the defining equations of $\game$ when $x$ ranges over $\mathbb{Z}_2$. 
    In particular, $\game$ is \cperf iff the defining equations of $\game$ have a solution over $\mathbb{Z}_2$. 
    
    In the quantum case, there are many possible strategies players can use for mapping questions to responses. One family of strategies that will be very important to this paper are called \emph{MERP strategies} in~\cite{watts2018algorithms}.
    As mentioned in the introduction, these are strategies in which each player holds one qubit of a 3 qubit GHZ state
    \begin{align}
        \ket{\GHZ} = \frac{1}{\sqrt{2}} \left(\ket{000} + \ket{111} \right)
    \end{align}
    and the $\alpha$-th player produces a response to question $i$ by measuring the observable 
    \begin{align}
    \exp(2 \pi i \sigma_z z^{(\alpha)}_i)  \sigma_x \exp(- 2 \pi i \sigma_z z^{(\alpha)}_i) \label{eq:MERP_strat}
    \end{align}
    on their qubit of the GHZ state. 
    In this discussion $\sigma_x, \sigma_z$ denote the Pauli X and Z operators, and we use multiplicative notation to describe responses, so a ``-1'' outcome corresponds to a 1 response in the standard (additive) notation, and a ``1'' outcome corresponds to a 0 response. The variables $z^{(\alpha)}_i$ depend on both the player $\alpha$ and the question $i$ they are sent and can take any value in $\mathbb{R}$.
    
    Straightforward calculation (see Claim 5.27 of \cite{watts2018algorithms}) shows that the score a MERP strategy achieves on a 3XOR game $\game$ with clauses $(a_i,b_i,c_i,s_i)$ is given by the function
    \begin{align}
    \label{eq:merpval}
        \frac{1}{2} + \frac{1}{2m} \sum_{i = 1}^m \cos\left(\pi \left(z^{(1)}_{a_i} + z^{(2)}_{b_i} +  z^{(3)}_{c_i} - s_i\right)\right) .
    \end{align}
    It follows that a 3XOR game $\game$ has a perfect MERP strategy iff there exists a real valued variable assignment for the $z^{(\alpha)}_i$ satisfying 
    \begin{align}
        z^{(1)}_{a_i} + z^{(2)}_{b_i} +  z^{(3)}_{c_i} = s_i \pmod{2}
    \end{align}
    for all clauses $(a_i,b_i,c_i,s_i)$ of $\game$. Equivalently, $\game$ has a perfect MERP stategy iff the defining equations of $\game$ have a solution over $\mathbb{R}$. Moreover, any real solution to the defining of $\game$ can be used to construct a MERP strategy simply by substituting the solution into \Cref{eq:MERP_strat}.
    
% {\bill Introduce real mod 2 game here= S}

    In~\cite{watts20203xor} it was shown that a 3 player XOR game $\game$ has a perfect quantum (or commuting operator) strategy iff the game has a perfect MERP strategy. The proof of this fact involved a system of equations ``dual'' to the defining equations of the game over $\mathbb{R}$, which had a solution iff $\game$ did not have a perfect quantum (or MERP) strategy. Because this dual system of equations could be solved in polynomial time,~\cite{watts20203xor} concluded that \qperf 3XOR games could be identified in polynomial time. However,~\cite{watts20203xor} did not give an algorithm for solving the defining equations of $\game$ over $\mathbb{R}$ directly.  
    We present such an algorithm in the next section.

\subsection{Efficient Algorithm for Characterizing \Qperf 3XOR Games}
\label{ssec:efficientAlgorithm}

We now give an algorithm which identifies \qcperf 3XOR games and finds
winning strategies for them. 
Of course, by the observations made in the previous section, \cperf games can easily be identified and classical strategies can be found via Gaussian elimination -- the primary reason we point out this feature of our algorithm is its possible use in theory. 

The algorithm makes use of a normal form for matrices called the \emph{Hermite Normal Form (HNF)}. For our purposes it suffices to know that the HNF is an efficiently computable normal form in which an integer valued matrix $M$ is rewritten as 
\begin{align}
    M = \Omega H
\end{align}
where $\Omega, H$ are also integer valued, $\Omega$ has integer inverse and $H$ is upper triangular with linearly independent rows(except for zero rows)~\cite{kannan1979polynomial}. A detailed discussion of the HNF is provided in \Cref{app:Hermite_Background}.

\begin{figure}[h]
\begin{algorithm}[H]
\caption{\QCperf Algorithm}
\label{alg:QandC_perf}

\begin{algorithmic}[1]

\Statex \textbf{Input:} An $n$ question, $m$ clause 3XOR game $\game$ with defining equation
\begin{align}
    \Gamma x = S  \ \ \ \pmod{2}.
\end{align}

\Statex \textbf{Output:} An classification of $\game$ as either \qperf, \qcperf, or neither. 
Additionally, if $\game$ is \qperf the algorithm outputs a vector $z \in \mathbb{R}^{3n}$ describing a perfect MERP strategy. If $\game$ is \cperf the algorithm outputs a vector $z \in \mathbb{Z}_2^{3n}$ describing a perfect deterministic strategy. 
\vspace{1pt}
\State
  \textbf{Procedure:} Compute the row Hermite Normal Form
   $(\Omega, H)$ of the matrix 
  $(\Gamma \ S)$  to obtain the integer-valued matrix
  \begin{align}
 H= \begin{pmatrix}
      R  & \bo \\
      0  & \bt  
 \end{pmatrix} 
 \end{align}
  where $\bo$ and $\bt$ are column vectors and the matrix $R$ is upper triangular and full rank.

% \item 
\State
Check whether all entries of $\bt$ are even integers.
\State If no, then $\game$ is neither \qperf nor \cperf. 
% \item
\State
If yes, then $\game$ is \qperf. Additionally:
 \begin{enumerate}[label = {(\alph*)}]
 \item 
 Vectors $z$ describing perfect MERP strategies are given by rational 
 (or real) solutions $z$ to the linear equation 
 \begin{align} \label{eq:alg_soln_eqn}
   Rz = \bo \pmod{2}.
 \end{align}
 Because $R$ is full rank and upper triangular at least one solution must exist (and is easy to compute). 
 \item The game $\game$ is \cperf iff \Cref{eq:alg_soln_eqn} has a binary solution 
 (and all entries of $\bt$ are even). If such a solution exists, 
 it describes a perfect deterministic strategy for the game.
\end{enumerate}

\end{algorithmic}
\end{algorithm}
\end{figure}

\begin{theorem}
\Cref{alg:QandC_perf} correctly characterizes all 3XOR games as either \Qperf, \Cperf, or not and produces perfect quantum and classical strategies for them whenever these strategies exist. Moreover it runs in time $\poly(n,m)$.
% , with runtime dominated by the cost of computing the HNF of the matrix $\begin{pmatrix}
%      \Gamma & S \end{pmatrix}$ 
\end{theorem}

\begin{proof}
Correctness and runtime of the algorithm follows directly from \Cref{thm:QperfHerm_app}.
\end{proof}

A sketch of how we implemented this algorithm (and a dual algorithm based on \cite{watts20203xor}) in Mathematica is in \Cref{sec:MmaPrograms}.

\section{Numerical Experiments}
\label{sec:Numerics_and_phase_transitions}

This section presents the main numerical results of this paper. 
Before presenting these results, we quickly review some background about phase transitions and how they relate to the \qperf and \cperf behavior of 3XOR games.

\subsection{Review of Phase Transitions and Pseudotelepathy in Random XOR Games}
\label{sec:phaseTransTelepathy}

Mathematically, a \textit{phase transition} (or \textit{threshold phenomenon}) refers to a phenomenon where qualitative properties of a large and randomly generated system change suddenly as a parameter controlling the probabilistic structure of the system undergoes a small change, that is, the parameter crosses a threshold.
This phenomenon was first observed by Erd\H{o}s and Renyi in their study of random of graphs \cite{erdos1960evolution}, and has since been studied in a wide variety of contexts, including statistical physics~\cite{yeomans1992statistical}, and theoretical computer science~\cite{friedgut2005hunting}. 

Relevant to this paper is the observation that many randomly generated systems of equations undergo a (sharp) satisfiability phase transition. Formally, this means there exists a threshold $T$ (a constant) with the property that randomly generated system of equations with $r$
variables and $m$ equations are almost-certainly satisfiable as $r \rightarrow \infty$ provided $m/r < T$ and almost-certainly unsatisfiable if $m/r > T$. In \cite{pittel2016satisfiability} a phase transition was shown for random linear systems of equations over $\mathbb{Z}_2$ with $k$ variables per equation ($k$-XOR-SAT). These systems of equations are closely related (but not identical to\footnote{The difference comes from the fact that the system of equations considered in \cite{pittel2016satisfiability} has 3 arbitrary variables per equation, while the defining equations of a 3XOR game have exactly one ``Alice'', ``Bob'', and ``Charlie'' variable per equation.}) the defining equations of 3XOR games. 

The close relationship between XOR games and linear systems of equations suggest that XOR games should undergo a \emph{classical satisfiability phase transition}: meaning that there exists a constant $C_c$ such that, in the limit of large $n$, XOR games with $n$ questions per player and $m$ clauses are almost-certainly \cperf if $m/n < C_c$ and almost certainly not \cperf if $m/n > C_c$. 

Prior to this work, little was known about whether or not XOR games also had a \qperf (or \emph{quantum satisfiability}) phase transition. In \Cref{sssec:heatmaps} we see strong evidence that they do.

The existence and location of a \qperf phase transition has important consequences concerning the likelihood of pseudotelepathy in random XOR games. Recall that an XOR game is called a pseudotelepathy game if it is \qperf but not \cperf. Because the quantum value of a game is bounded below by its classical value, we know that random generated $m$ question $n$ clause 3XOR games are \qperf with high probability when $m/n < C_c$. Now, if these games did not undergo a quantum satisfiability phase transition, or underwent a phase transition at some constant $C_q > C_c$ random $n$ question $m$ clause 3XOR games would still almost certainly be \qperf when $m/n = C_c + \epsilon$ for some small constant~$\epsilon$. These games would be almost certainly \qperf and almost certainly not \cperf, and therefore be pseudotelepathy games with probability approaching~1 in the limit of large~$n$. 

In the next section, we provide numerical evidence that the pseudotelepathy probability for all values of $n$ and $m$ is bounded well below 1. Thus, the evidence suggests that the scenario described in the paragraph above \emph{does not} occur, and hence 3XOR games undergo a quantum satisfiability phase transition at the same location as the classical satisfiability phase transition.

\subsection{Results}

\subsubsection{Heatmaps of Phase Transition}
\label{sssec:heatmaps}

We begin with figures showing how the probability of randomly generated 3XOR games being quantum or classically perfect varies with clause and question number. \Cref{fig:heatC} illustrates the classical situation, while \Cref{fig:heatQ} shows the quantum case. In both cases, we see that as question number increases the rate (in terms of ratio of clauses to questions) at which games transition from being very-likely-not-perfect (dark coloured) to very-likely-perfect
(light coloured) increases. 
This suggests the existence of a sharp phase transition in both the classical and quantum cases. We also see the locations of these two phase transitions approximately coincide. \Cref{fig:heatPseudoTel} gives the probability a game is a pseudotelepathy game (equal to the probability a game is \qperf minus the probability it is \cperf) as a function of clause and question numbers. 

\begin{figure*}[hp]
\centering
\subcaptionbox{\Cperf Probability\label{fig:heatC}}{
    \centering
    \includegraphics[width = 0.39\textwidth]{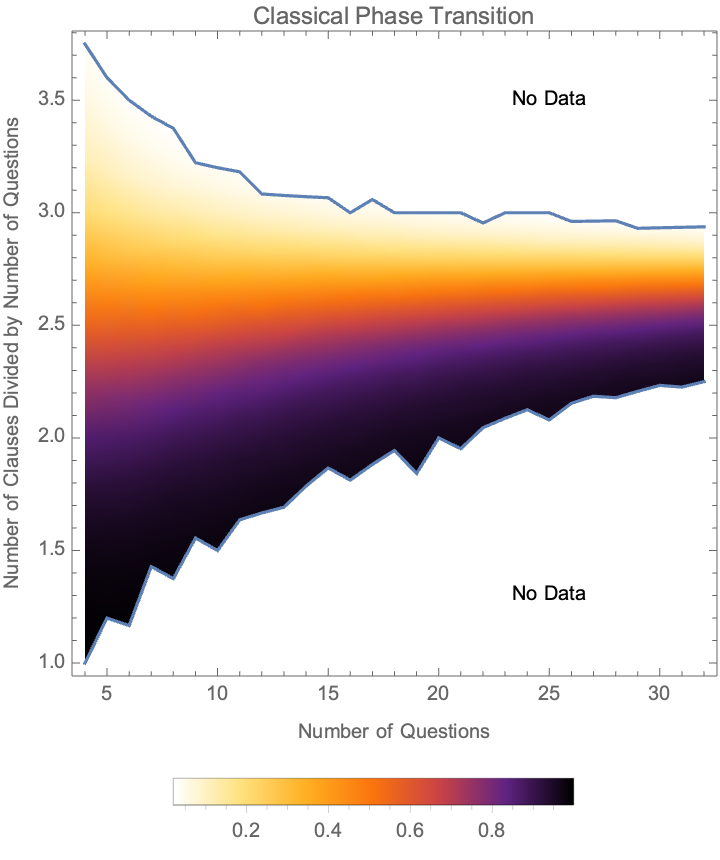}

}
\hspace{8pt}
\subcaptionbox{\Qperf Probability\label{fig:heatQ}}{
    \centering
    \includegraphics[width = 0.39\textwidth]{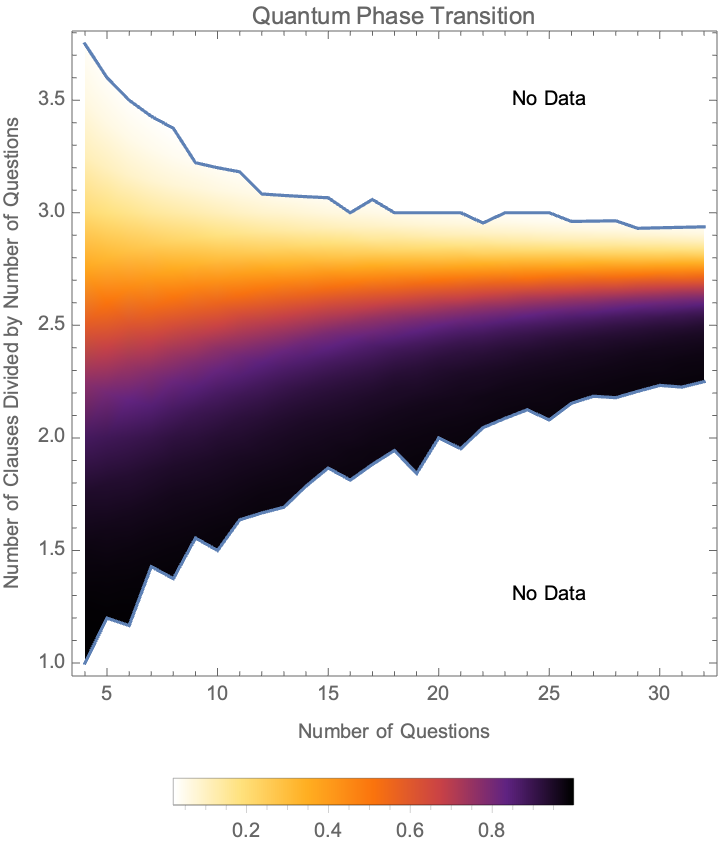}
}
\subcaptionbox{Pseudotelepathy Probability\label{fig:heatPseudoTel}}{
    \centering
    \includegraphics[width = 0.5\textwidth]{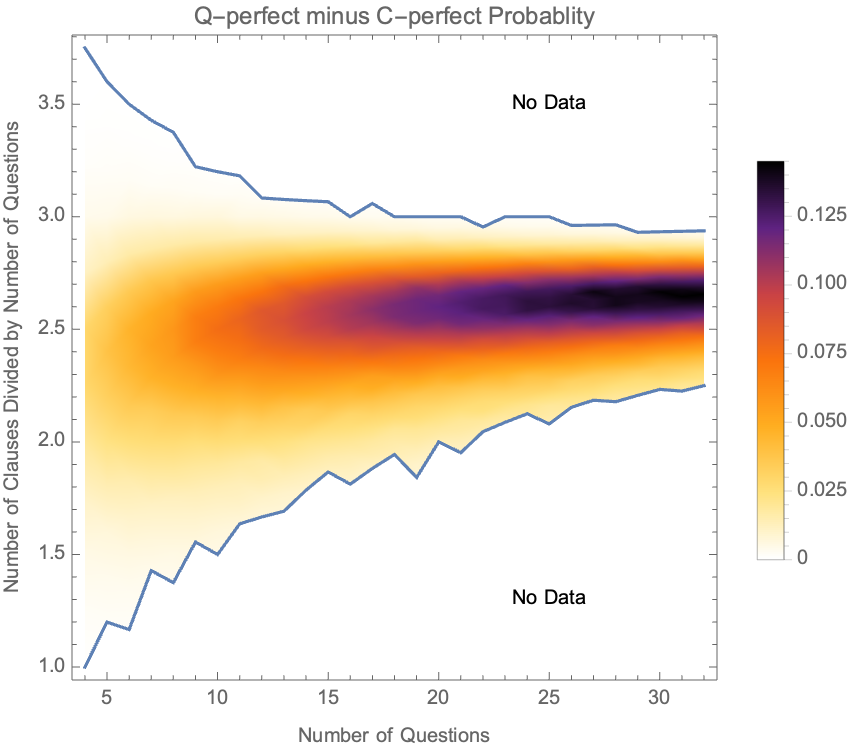}
}
\captionsetup{width = 0.95\textwidth}
\caption{The probabilities of a randomly generated 3XOR game being \Cperf, \Qperf, or pseudotelepathic as a function of number of clauses ($m$) and questions ($n$). The $x$-axis labels the number of questions while the y-axis shows a ratio $m/n$.
Probabilities are indicated by color, with dark colors being more likely and light colors less likely.
Data is generated using 50K samples per question and clause number.
}
\end{figure*}

\subsubsection{Phase Transition at Fixed Question Numbers}
\label{sssec:phase_transitions}

\Cref{fig:heatC,fig:heatQ} illustrate the \qcperf phase transitions, but are ``blurry''.
To get more detailed information one can plot slices of them. 
We illustrate with \Cref{fig:8Qcross_section,fig:38Qcross_section,fig:100Qcross_section} which plot vertical slices taken at different fixed question numbers. These experiments require less information than needed for heat maps, thus we are able to handle more questions than in the heatmaps.

\begin{figure*}[hp]
\centering
\subcaptionbox{8 Questions \label{fig:8Qcross_section}}{
    \centering
    \includegraphics[width = 0.42\textwidth]{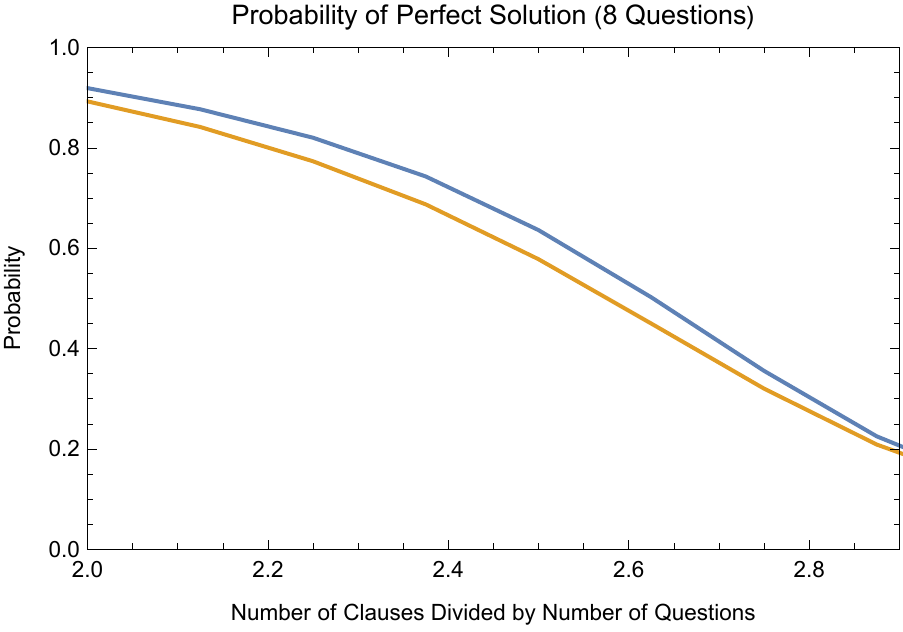}
}
\hspace{8pt}
\subcaptionbox{38 Questions \label{fig:38Qcross_section}}{
    \centering
    \includegraphics[width = 0.42\textwidth]{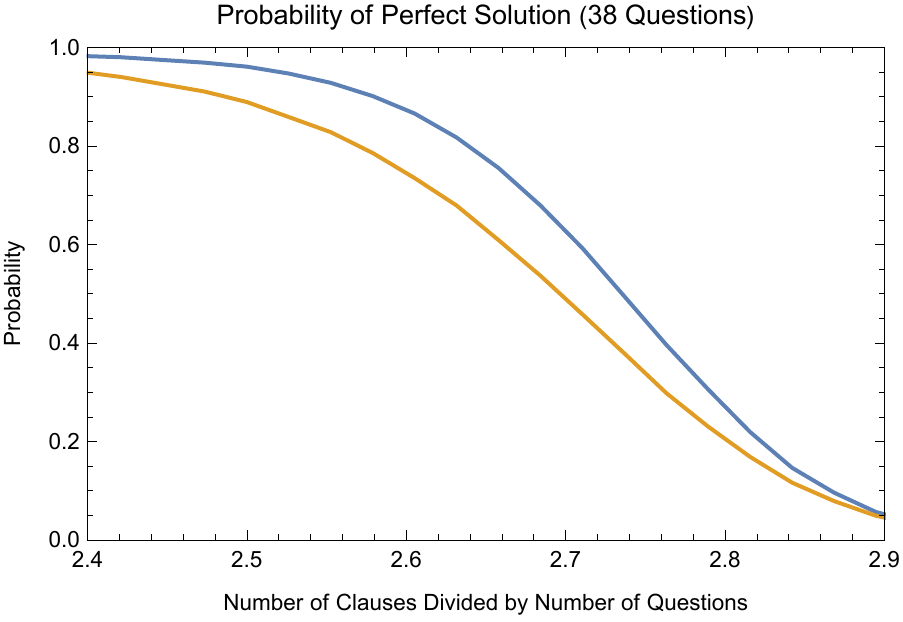}
}
\hspace{8pt}
\subcaptionbox{100 Questions \label{fig:100Qcross_section}}{
    \centering
    \includegraphics[width = 0.42\textwidth]{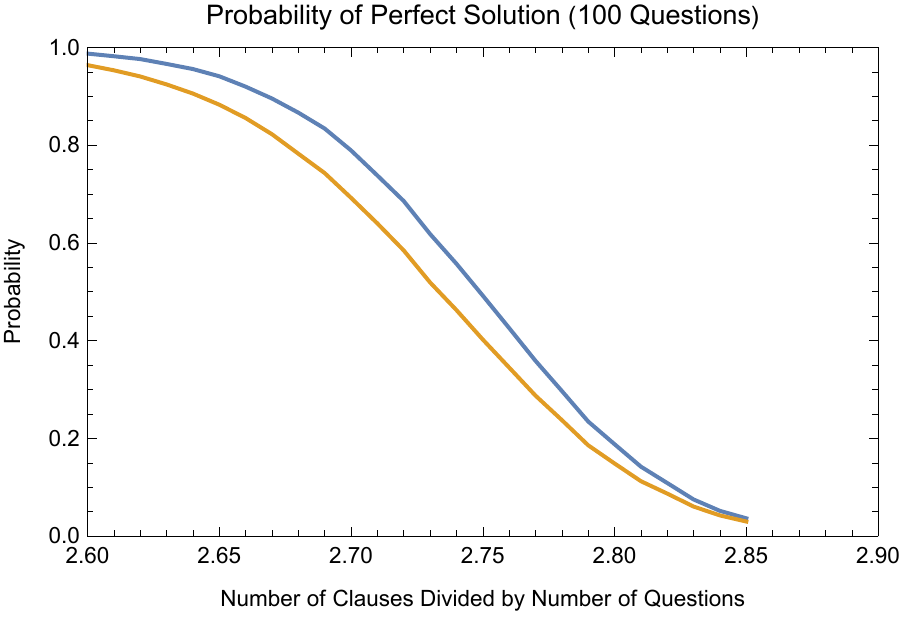}
}
\captionsetup{width=0.98\textwidth}
\caption{The Quantum/Classical Probability that a random 3XOR game at three fixed question numbers (8, 38, 100) is Q-perfect (top curve) and C-perfect (bottom curve) as a function of the number of clauses in the game.
Sample size is 50K.}
\end{figure*}

\subsubsection{Probability of Pseudotelepathy Games}
\label{sssec:pseudotelepathy_prob}

Finally, we investigate the maximum probability at each fixed question number $n$ that a randomly generated game with $n$ questions and $m$ clauses is a pseudotelepathy game. In \Cref{fig:Max_Pseudotel_Posn} we plot as a function of $n$ the number of clauses $m$ which maximizes the probability that a random game with $n$ questions and $m$ clauses is a pseudotelepathy game. In \Cref{fig:Max_Pseudotel_Prob} we plot this maximized probability as a function of $n$. 

\begin{figure*}[hp]
\centering
\subcaptionbox{Max Pseudotelepathy Position \label{fig:Max_Pseudotel_Posn}}{
    \centering
    \includegraphics[width =
    0.42\textwidth]{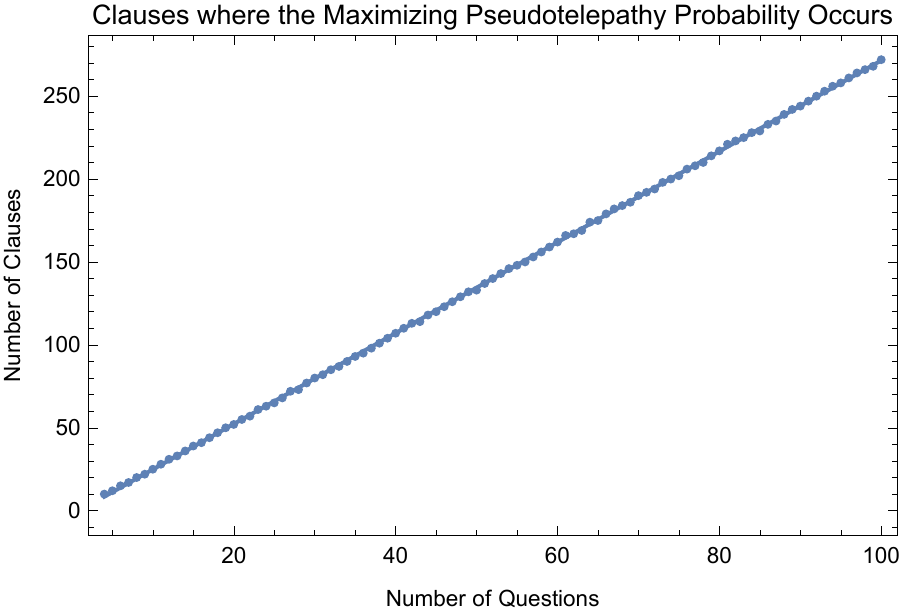}
}
\hspace{8pt}
\subcaptionbox{Max Pseudotelepathy Probability \label{fig:Max_Pseudotel_Prob}}{
    \centering
    \includegraphics[width = 0.42\textwidth]{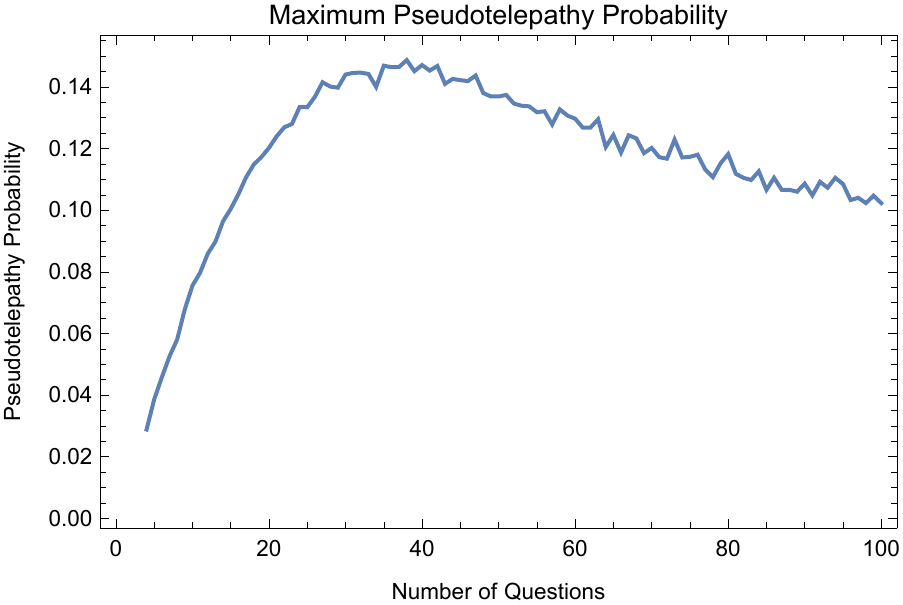}
}
\captionsetup{width=0.98\textwidth}
\caption{An investigation of the maximum probability of a random XOR game being pseudotelepathic as a function of question number. Figure (a) gives the number of clauses which maximizes this probability at each fixed question number. We also give a linear fit. The equation of the fit is 
    $$m = - 2.54013 +2.7405 n .$$
Figure (b) gives the probability that a random gave is pseudotelepathic at fixed question number when the number of clauses is set to the maximizing value given in figure (a). For both figures sample size is 50K for $n\leq 32$ and 10K for $n > 32$.
}
\end{figure*}

\FloatBarrier

\section{Discussion}

\label{sec:Discussion}

In this section we summarize the main conclusions of this paper and suggest avenues for future work.  \\

\noindent 
 {\bf Pseudotelepathy and Phase Transitions in Random 3XOR Games: } 
The numerical section of this paper focused on the probability that a randomly generated $n$ question, $m$ clause 3XOR game was a pseudotelepathy game. We now denote this probability $\rho(n,m)$. Our main observations were as follows: 
\begin{enumerate}
    \item At all fixed values of $n$ (except for $n$ very small), the maximum pseudotelepathy probability
    \begin{align}
        \mu(n) := \max_m \rho(m,n)
    \end{align}
    occurs when $m/n \approx 2.74$ (\Cref{fig:Max_Pseudotel_Posn}).
    \item The overall maximum pseudotelepathy probability appears bounded 
    \begin{align}
        \max_n \mu(n) \lesssim 0.14
    \end{align}
    by a maximum occurring at $n \approx 38$, and falls off at larger values on $n$ (\Cref{fig:Max_Pseudotel_Prob}).
\end{enumerate}
Based on these observations and  \Cref{fig:heatQ,fig:heatC} we make the following conjecture:
\begin{enumerate}[resume]
    \item \label{item:conjecture}
    Conjecture: XOR games have a \qperf phase transition and its location $C_q$ coincides with the location $C_c$ of the \cperf phase transition: $C_q = C_c \approx 2.74$.
\end{enumerate}
  Indeed (by the argument near the end of \Cref{sec:phaseTransTelepathy}), this conjecture is true unless $\mu(n)$ increases to 1.
  We suggest proving the the conjecture above as an open question for future work. \\

\noindent
{\bf Perfect Games Algorithm and MERP Strategies:} \Cref{alg:QandC_perf} gives an efficient method of identifying 3XOR games with perfect quantum strategies, and describing those strategies. It allows the classification of games with hundreds of questions and clauses -- a regime far beyond techniques such NPA or exhaustive search.

We remark here that \Cref{alg:QandC_perf} also extends naturally beyond 3XOR games to identify $k$-XOR games with perfect MERP strategies. When $k > 3$ games may have perfect quantum strategies without perfect MERP strategies (\cite{watts2018algorithms}), so this algorithm no longer classifies all \Qperf games. Still, we hope that it will serve as a starting point for future insights about $k$-XOR games.\\

\noindent
{\bf Comparison Problems:}
Finally, we note this paper introduces what seems to be a new class of mathematical questions. Consider a equation of the form 
\begin{align}
    \Gamma x = S
\end{align}
where $\Gamma \in \mathbb{Z}_2^{m \times 3n}$ and $S \in \mathbb{Z}_2^{m \times 1}$ are generated randomly, perhaps with special structure (as with 3XOR). Questions about the satisfiability of  this system of equations over $\mathbb{Z}_2$ are common in computer science cf.~\cite{friedgut2005hunting}.
The crux of this paper is comparing this to satisfiability of the system of equations  
\begin{align}
    \Gamma x = S \pmod{2}
\end{align}
over $\mathbb{R}$. In the examples considered in this paper we found key probabilistic behavior of the randomly generated equations (for example the existence and location of the satisfiability phase transition) remained consistent when the equations were studied over $\mathbb{Z}_2$ and $\mathbb{R}$. Further understanding this phenomenon -- as well as characterizing when the behavior of the randomly generated equations differs -- seems to be a mathematically rich question deserving of further study.

\bigskip

\begin{acknowledgments}
We are grateful to CMSA for facilitating the Collaboration of the first two authors.
Thanks to Zinan Hu and Jared Hughes for help with calculations
and to Kiran Kedlaya for discussions about Diophantine equations. We also wish to thank Daniel Lichtbau from Wolfram Research for assistance with Mathematica. ABW was supported in part by NSERC DG 2018-03968.
\end{acknowledgments}

%\printbibliography
\bibliography{ref}% Produces the bibliography via BibTeX.

\newcommand{\etalchar}[1]{$^{#1}$}
\begin{thebibliography}{WHKN18}

\bibitem[BBT05]{brassard2005quantum}
Gilles Brassard, Anne Broadbent, and Alain Tapp.
\newblock Quantum pseudo-telepathy.
\newblock {\em Foundations of Physics}, 35(11):1877--1907, 2005.

\bibitem[Bel64]{bell1964einstein}
John~S Bell.
\newblock On the einstein podolsky rosen paradox.
\newblock {\em Physics Physique Fizika}, 1(3):195, 1964.

\bibitem[big18]{big2018challenging}
Challenging local realism with human choices.
\newblock {\em Nature}, 557(7704):212--216, 2018.

\bibitem[BLS21]{birmpilis2021fast}
Stavros Birmpilis, George Labahn, and Arne Storjohann.
\newblock A fast algorithm for computing the smith normal form with multipliers
  for a nonsingular integer matrix.
\newblock {\em arXiv preprint arXiv:2111.09949}, 2021.

\bibitem[CCC93]{cohen1993course}
Henri Cohen, Henry Cohen, and Henri Cohen.
\newblock {\em A course in computational algebraic number theory}, volume~8.
\newblock Springer-Verlag Berlin, 1993.

\bibitem[Eke91]{QKD2:ekert1991quantum}
Artur~K Ekert.
\newblock Quantum cryptography based on bell’s theorem.
\newblock {\em Physical review letters}, 67(6):661, 1991.

\bibitem[ER{\etalchar{+}}60]{erdos1960evolution}
Paul Erdos, Alfr{\'e}d R{\'e}nyi, et~al.
\newblock On the evolution of random graphs.
\newblock {\em Publ. Math. Inst. Hung. Acad. Sci}, 5(1):17--60, 1960.

\bibitem[Fri05]{friedgut2005hunting}
Ehud Friedgut.
\newblock Hunting for sharp thresholds.
\newblock {\em Random Structures \& Algorithms}, 26(1-2):37--51, 2005.

\bibitem[JNV{\etalchar{+}}20]{ji2020mip}
Zhengfeng Ji, Anand Natarajan, Thomas Vidick, John Wright, and Henry Yuen.
\newblock {MIP}$^*$={RE}.
\newblock {\em Preprint}, 2020.
\newblock \url{https://arxiv.org/abs/2001.04383}.

\bibitem[KB79]{kannan1979polynomial}
Ravindran Kannan and Achim Bachem.
\newblock Polynomial algorithms for computing the smith and hermite normal
  forms of an integer matrix.
\newblock {\em siam Journal on Computing}, 8(4):499--507, 1979.

\bibitem[PS16]{pittel2016satisfiability}
Boris Pittel and Gregory~B Sorkin.
\newblock The satisfiability threshold for k-xorsat.
\newblock {\em Combinatorics, Probability and Computing}, 25(2):236--268, 2016.

\bibitem[RHH{\etalchar{+}}18]{rauch2018cosmic}
Dominik Rauch, Johannes Handsteiner, Armin Hochrainer, Jason Gallicchio,
  Andrew~S Friedman, Calvin Leung, Bo~Liu, Lukas Bulla, Sebastian Ecker, Fabian
  Steinlechner, et~al.
\newblock Cosmic bell test using random measurement settings from high-redshift
  quasars.
\newblock {\em Physical review letters}, 121(8):080403, 2018.

\bibitem[RUV13]{DQC:reichardt2013classical}
Ben~W Reichardt, Falk Unger, and Umesh Vazirani.
\newblock Classical command of quantum systems.
\newblock {\em Nature}, 496(7446):456--460, 2013.

\bibitem[Sch98]{schrijver1998theory}
Alexander Schrijver.
\newblock {\em Theory of linear and integer programming}.
\newblock John Wiley \& Sons, 1998.

\bibitem[SL96]{storjohann1996asymptotically}
Arne Storjohann and George Labahn.
\newblock Asymptotically fast computation of hermite normal forms of integer
  matrices.
\newblock In {\em Proceedings of the 1996 international symposium on Symbolic
  and algebraic computation}, pages 259--266, 1996.

\bibitem[WH20]{watts20203xor}
Adam~Bene Watts and J.~William Helton.
\newblock 3{XOR} games with perfect commuting operator strategies have perfect
  tensor product strategies and are decidable in polynomial time.
\newblock {\em Preprint}, 2020.
\newblock \url{https://arxiv.org/abs/2010.16290}.

\bibitem[WHKN18]{watts2018algorithms}
Adam~Bene Watts, Aram~W Harrow, Gurtej Kanwar, and Anand Natarajan.
\newblock Algorithms, bounds, and strategies for entangled {XOR} games.
\newblock {\em Preprint}, 2018.
\newblock \url{https://arxiv.org/abs/1801.00821}.

\bibitem[Yeo92]{yeomans1992statistical}
Julia~M Yeomans.
\newblock {\em Statistical mechanics of phase transitions}.
\newblock Clarendon Press, 1992.

\end{thebibliography}

\newpage 

\def\bbZ{\mathbb Z}

\section{Appendix: MERP Strategies}
\label{app:duality}

This appendix gives a more detailed overview of MERP strategies. It repeats some parts of \Cref{sec:XOR_games_and_algorithm} in order to be nearly  self contained.

We consider a 3 player XOR game $\game$ with clauses $(a_i,b_i,c_i,s_i)$ for $i \in [m]$ and defining equations
\begin{align}
    \Gamma z = S \pmod{2}
\end{align}

\ssec{MERP strategies for 3XOR}
A MERP strategy for such a game is a strategy in which players share a 3 qubit GHZ state and player $\alpha$ produces a response to question $i$ by measuring the observable 
\begin{align}
    \exp(2 \pi i \sigma_z z^{(\alpha)}_i)  \sigma_x \exp(- 2 \pi i \sigma_z z^{(\alpha)}_i)
\end{align}
applied to their qubit of the GHZ state. 
In this discussion $\sigma_x, \sigma_z$ denote the Pauli X and Z operators, and we are using multiplicative notation to describe responses, so a ``-1'' outcome corresponds to a 1 response in the standard (additive) notation, and a ``1'' outcome corresponds to a 0 response. The variables $z^{(\alpha)}_i$ can take any value in $\mathbb{R}$ and depend on both the player $\alpha$ and the question $i$ they are sent. 

If players are sent questions $(i_1, i_2, i_3)$ the expected value of the product of their responses (again using multiplicative notation) is given by the observable
\begin{align}
    \exp(2 \pi i \sigma_x z^{(1)}_{i_1}) \otimes \exp(2 \pi i \sigma_x z^{(2)}_{i_2}) \otimes \exp(2 \pi i \sigma_x z^{(3)}_{i_3})
\end{align}
applied to the GHZ state. 

Straightforward calculation (see Claim 5.27 of \cite{watts2018algorithms}) then shows that the score a MERP strategy achieves on the game $\game$ is given by the function
\begin{align}
\label{eq:merpval2}
    \frac{1}{2} + \frac{1}{2m} \sum_{i = 1}^m \cos\left(\pi \left(z^{(1)}_{a_i} + z^{(2)}_{b_i} +  z^{(3)}_{c_i} - s_i\right)\right) .
\end{align}
This implies the next lemma:
\begin{lemma}
An 3XOR game with clauses $(a_i,b_i,c_i,s_i)$ for $i \in [m]$ has a perfect MERP strategy iff there exist variables $z^{(\alpha)}_{j} \in \mathbb{R}$ satisfying the equations
\begin{align}
    z^{(1)}_{a_i} + z^{(2)}_{b_i} +  z^{(3)}_{c_i} = s_i \pmod{2}
\end{align}
for all $i \in [m]$, or, equivalently, 
if the defining equations of $\game$
\begin{align} \label{eq:merp_defeq_app}
  \Gamma  z
    = 
    S \pmod{2}.
\end{align}
have a solution with $z$ taking values in $\mathbb{R}^{3n}$. 

\end{lemma}

The main result of \cite{watts20203xor} shows that a 3XOR game has a perfect quantum strategy  iff the 3XOR game has a perfect MERP strategy.\footnote{Actually, the main result of \cite{watts20203xor} shows a 3XOR game has commuting operator value 1 iff the game has a perfect MERP strategy -- a strictly stronger statement.} This implies an immediate strengthening of the lemma above:
\begin{corollary}
\label{cor:perfect_3XOR_merp}
    An 3XOR game with clauses $(a_i,b_i,c_i,s_i)$ for $i \in [m]$ has a perfect quantum strategy iff there exist variables $z^{(\alpha)}_{j} \in \mathbb{R}$ satisfying the equations 
\begin{align}
    z^{(1)}_{a_i} + z^{(2)}_{b_i} +  z^{(3)}_{c_i} = s_i \pmod{2}
\end{align}
for all $i \in [m]$, or, equivalently, 
if the defining equations of $\game$ 
\begin{align}
    \label{eq:merp_matrix_SOE}
  \Gamma z 
    = 
    S \pmod{2}.
\end{align}
have a solution with $z$ taking values in $\mathbb{R}^{3n}$. 
\end{corollary}

\subsection{MERP for kXOR}
A MERP strategy for kXOR has the same form as above but tensored
with an appropriate number of $I_2$.
Also the value of a MERP strategy is given by the expected  natural generalization of \Cref{eq:merpval} hence a \qperf
MERP strategy again  corresponds to linear equations 
in analogy with \Cref{eq:merp_defeq_app} with solutions obtainable via \Cref{thm:QperfHerm_app}.

\section{Appendix: Solving the \perfeq}

In this section we prove correctness of two algorithms for solving the defining equations of a game over $\mathbb{Z}_2$ and $\mathbb{R}$.
The first is based on the Smith Normal Form (SNF) of a matrix and gives a clear comparison between \qperf and \cperf games. The second algorithm is faster, though perhaps less conceptually clear. It is based on the Hermite Normal Form (HNF) of a matrix.

\subsection{A MERP solution from the Smith Normal Form}
\label{sec:merpLinSoln}

\def\RR{{\mathbb R}}
\def\bbZ{{\mathbb Z}}

Here we give an algorithm for solving \Cref{eq:merp_matrix_SOE} over $\mathbb{R}^{3n}$ and $\mathbb{Z}_2^{3n}$. This algorithm determines whether or not a 3XOR game has a \qperf or \cperf strategy and provides a description of this strategy when it exists.

We begin by recalling the some properties of the Smith Normal Form of a matrix. The Smith Normal Form of an $m \times 3n$ integer valued matrix $M$ is given by a 3-tuple of integer valuled matrices $(\Omega, D, \Psi)$ where
\begin{enumerate}
\item $M = \Omega D \Psi$
\item 
$\Omega$ is $m \times m$ with $ det = \pm 1$. 
\item 
$\Psi$ is $3n \times 3n$ with $det=\pm 1$. 
\item 
$D = \text{Diag}(d_1, d_2, \dots d_r, 0, \dots, 0)$ is $m \times 3n$.
The $d_j$ are called the invariants of $\Gamma$
\item
For each invariant $d_i$, $d_i$ divides $d_{i+1}$.
\end{enumerate}
The Smith Normal form of a matrix always exists and can be computed in time polynomial in the size of the matrix and the bit complexity of the largest integer entry~\cite{kannan1979polynomial}. Finding fast algorithms for computing the SNF of a matrix remains an active research area~\cite{birmpilis2021fast}. 

Next, we show that we can use the Smith Normal Form of the matrix $\Gamma$ to easily determine whether or not the defining equations of a 3XOR game have real or binary solutions, and compute those solutions when they exist. 

\begin{theorem}
 Let $\game$ be a 3XOR game with defining equation 
 \begin{align}
     \Gamma z = S \pmod{2}
 \end{align}
 and let $(\Omega, D, \Psi)$ be the Smith Normal Form of the matrix $\Gamma$. Then
 \begin{enumerate} 
 \item
 The game $\game$ has a perfect quantum strategy iff 
 $$   (\Omega^{-1} S)_j \text{ is even}   \quad  \text{ whenever } d_j =0$$
in which case any real solution to the defining equations of $\game$ satisfies $z= \Psi^{-1} y$ with $y_i$ a rational number
given by
\begin{align}
\label{eq:solLinMERP}
     y_i = \frac 1 d_i ( \Omega^{-1} S + 2 {\mathbb N})_i \quad \text{ whenever }  d_i \not =0,
\end{align}

 \item 
 The game $\game$ has a perfect  classical strategy, iff
 $$ (\Omega^{-1}S)_j \text{ is even } \quad \text{ whenever } d_j \text{ is even.} $$
%  for no $i$ is
%  $d_i$  even  
%  while  $(\Omega^{-1}S)_i $ is odd.
\noindent In which case a binary solution to the defining equations is given by $z = \Psi^{-1} \mathbf{1}$ where $\mathbf{1}$ is the all ones vector and $\Psi^{-1}$ now denotes the inverse of $\Psi^{-1}$ over the field~$\mathbb{F}_2$. 
 
 \end{enumerate}
\end{theorem}

\begin{proof}

Starting with the Smith Decomposition and letting $\mathbb{N}$ denote an arbitrary vector of integers we can rewrite the defining equations of $\game$ as 
\begin{align}
     \Omega D y := \Omega D \Psi z= S  \pmod{2}.
\end{align}
Where we defined the vector $y := \Psi z$. Because $\Omega$ is invertible with integer inverse ($\det(\Omega) = \pm 1$) this equation is equivalent to 
\begin{align}
    D y = \Omega^{-1} S \pmod{2}. \label{eq:definingEquationSNF}
\end{align}
Now it is clear a solution with $y \in \mathbb{R}^{3n}$ exists iff 
\begin{align}
    (\Omega^{-1} S)_j \text{ is  even }   \quad  \text{ whenever } d_j =0
\end{align}
and when this solution exists it is given by \Cref{eq:solLinMERP}. Noting that $\Psi$ is invertible and setting $z = \Psi^{-1} y$ shows that this solution also implies a real solution to the defining equations of $\game$, as claimed.

To understand binary solutions to the defining equations of the game $\game$ note that \Cref{eq:definingEquationSNF} has a solution with an integer (or equivalently, binary) vector $y$ iff 
\begin{align}
   (\Omega^{-1} S)_j =  d_j \pmod{2}
\end{align}
for all $j$, in which case taking $y = \mathbf{1}$ gives a solution. Then letting $\Psi^{-1}$ denote the inverse of $\Psi$ over the field $\mathbb{F}_2$ and taking $z = \Psi^{-1} y = \Psi^{-1} \mathbf{1}$ gives a binary solution to the defining equations, as desired. 

\end{proof}
 
\subsection{A MERP solution from Hermite normal form}

The row-style Hermite Normal Form (HNF) of an integer matrix is 
the version of its row reduction which preserves integers all matrix entries
being integers. It is easier to compute than the SNF of a matrix. In this section we show it is easy to solve the defining equations of an XOR game of $\mathbb{R}$ once the matrix $\begin{pmatrix} \Gamma & S \end{pmatrix}$ has been put in HNF.

\subsubsection{Hermite Background}

\label{app:Hermite_Background}

We begin by reviewing the definition of the Hermite Normal Form of a matrix. The HNF of an integer valued $d_1 \times d_2$ matrix $M$ is defined by a tuple of matrices $(\Omega, H)$ with $\Omega$ having dimensions $d_1 \times d_1$, $H$ having dimensions $d_1 \times d_2$ and $(\Omega, H)$ satisfying the restrictions (adapted from Definition 2.4.2 of~\cite{cohen1993course}):
\begin{enumerate}

\item  $M=\Omega H$.

\item $\det(\Omega) = \pm 1$.

\item Any zero rows of $H$ occur at the bottom of $H$.

\item For any nonzero row of $H$ the first nonzero entry, called a pivot, occurs at position $h_{i, f(i)}$ where $f(i)$ is a strictly increasing function of $i$. 

\item For any nonzero row $i$ we have $h_{i,f(i)} \geq 1$ and $0 \leq h_{i',f(i)} < h_{i, f(i)}$ for all $i' < i$.

\end{enumerate}

The Hermite Normal Form of an integer valued matrix matrix always exists, and can be computed in time polynomial in the size of the matrix and the bit complexity of the largest integer entry of the matrix~\cite{kannan1979polynomial,storjohann1996asymptotically}. 
 
 \def\bbZ{{\mathbb Z}}
 \def\cE{{\mathcal E}}
 \def\eps{{\varepsilon}}
 
 \subsubsection{A Hermite NF algorithm}
 
  The next theorem gives a way to check if a 3XOR game $\game$ has a \qperf or \cperf strategy, and to find these strategies when they exist. We also show this algorithm characterizes when games have multiple distinct perfect MERP strategies. (Though we note this condition doesn't relate immediately to self-testing, both because it only considers MERP strategies, and because two MERP strategies which involve different rotation angles and thus appear ``distinct'' may be equivalent after isometry).

 \begin{theorem}
 \label{thm:QperfHerm_app}
 Let $\game$ be a 3XOR game with defining equation 
 \begin{align}
       \Gamma z = S \pmod{2}
 \end{align}
  Denote the row Hermite Normal Form
  of the matrix $(\Gamma \ S)$  by $(\Omega, H)$ with
  \begin{align}
  \label{eq:hermDec}
 H= \begin{pmatrix}
      R  & \bo \\
      0  & \bt  
 \end{pmatrix} 
 \end{align}
  where $\bo$ and $\bt$ are column vectors and the matrix $R$ contains all the rows of $H$ with nonzero pivots except for the one possibly contained in $\bt$. Only the first entry of $b^2$ can be non zero.\footnote{This follows immediately 
 from the definition of the HNF.}

  Then the defining equations of $\game$:
    \begin{enumerate}
  \item
  \label{it:QperfHerm}
  have a rational (or real) solution
  iff all entries of 
  the vector 
$\bt$ are  even entries. In this case, the solution is given by any vector $z$ which solves the equation 
 \begin{align}
   Rz = \bo \pmod{2}.
 \end{align}
 Thus $\bt$ being an even vector is equivalent to the game 
$\game$ being \qperf.

\item 
 \label{it:CperfHerm}
 has a binary solution iff the vector $\bt$ has even entries and the equation 
\begin{align}
    R z = \bo \pmod{2}
\end{align}
has a solution with $z \in \mathbb{Z}_2$.
\end{enumerate}
%%%
Finally, we note that this characterization lets us determine when $3XOR$ games have unique MERP solutions. 
\begin{enumerate}[resume]
  \item
  \label{it:QperfHermUniq}
 A 3XOR game has a unique perfect MERP
 strategy iff $\bt$ is an even vector
 and $R$ is a square $3n \times 3n $
 matrix with $\det(R) = \pm 1 $.
 Since $R$ is triangular with integer entries  $\det(R) = \pm 1 $
 is equivalent to each diagonal entry
 being $\pm 1$.
 \end{enumerate}
  
 \end{theorem}

\begin{proof}
We begin by rewritting the determining equation of the game $\game$ 
\begin{align}
    \Gamma z = S \pmod{2}
\end{align}
as 
\begin{align}
    \begin{pmatrix}
         \Gamma &
         S
    \end{pmatrix} 
    \begin{pmatrix}
         z \\
         -1
    \end{pmatrix}
    = 0 \pmod{2}.
\end{align}
Now, using the transformation of $(\Gamma \ S)$ to Hermite normal form this equation becomes
\begin{align}
\label{eq:GSz}
    \Omega H \
    \begin{pmatrix}
         z \\
         - 1
    \end{pmatrix} 
    & = 0 \pmod{2} \\
    \Leftrightarrow \ \  
  \begin{pmatrix}
      R & \bo \\
      0 & \bt 
 \end{pmatrix}
    \begin{pmatrix}
         z \\
         - 1
    \end{pmatrix}   & = 0 \pmod{2}. \label{eq:defEqHNF}
\end{align}

We can now prove \Cref{it:QperfHerm,it:CperfHerm,it:QperfHermUniq}.

\begin{enumerate}
    \item Clearly, if a solution to \Cref{eq:defEqHNF} exists, then $\bt$ must be an even vector. Conversely, suppose $\bt$ is even. Then \Cref{eq:defEqHNF} becomes 
    \begin{align}
        R z\ -\bo \  = 0 \pmod{2}.  \label{eq:def_eq_HNF_btzero}
    \end{align}
    
    Since $R$ applied to $\mathbb{R}^{3n} $ has range onto we may take $z \in \mathbb{R}^{3n}$ to make $R z - \bo$ an even vector, giving a real solution to the defining equations. Since all the coefficients of the defining equations are integer a real solution exists iff a rational solution exists.
    
    \item By the same logic as (1), if $\bt$ is even then a binary vector $z$ solves the defining equations of $\game$ iff
    \begin{align}
       R z\ -\bo \  &= 0 \pmod{2} \\
        \Leftrightarrow Rz &= \bo \pmod{2}. 
    \end{align}
    
    \item 
    First note two real vectors $z_1$ and $z_2$ correspond to different MERP strategies iff $z_1 \neq z_2 \pmod{2}$. Then, assuming $\bt$ is even, the game $\game$ has two distinct MERP strategies iff there exist vectors $z_1, z_2$ in $\mathbb{R}^{3n}$ with 
    \begin{align}
         Rz_1 = Rz_2 = \bo \pmod{2}\label{eq:unique1}
    \end{align}
    but 
    \begin{align}
        z_1 \neq z_2 \pmod{2}. \label{eq:unique2}
    \end{align}
    So a game has two distinct perfect MERP strategies iff
    \begin{align}
        \label{eq:Rker}
        R(z_1 - z_2 ) = 0  \pmod{2}.
    \end{align}
    for two vectors $z_1 \neq z_2 \pmod{2}$.
    
    We first show that a 3XOR game $\game$ has a unique perfect MERP strategy only if $R$ has no kernel. Assume for contradiction that $\game$ has a unique MERP strategy and that $R$ has a kernel containing the vector $\mu \in \mathbb{R}^{3n}$.
    Then there is a $w \in \bbZ$ for which
    $\mu/2^w$ is not even but is in the kernel of $R$. Given any vector $z_1$ satisfying $Rz_1 = \bt$ we also have $R(z_1 + \mu/2^w) = \bt$, contradicting uniqueness by \Cref{eq:Rker}.

    We just showed 
    $R$ has no kernel and 
    also it has range onto, so $R$ has an
inverse over $\mathbb{R}$, which we denote~$R^{-1}$.
Let $\mathcal{E}$ denote the set of even vectors.
Uniqueness, thanks to 
\Cref{eq:Rker,eq:unique2}, is equivalent to the condition 
    $R^{-1} \cE \to \cE$, that is $R^{-1}$ maps even vectors to even vectors. 
    Dividing everything by $2$, this is equivalent to the condition
    \begin{align}
         R^{-1} \bbZ\to \bbZ
    \end{align}
    which happens iff all entries of $R^{-1}$ are integers (a non integer entry in the $j^{th}$ column can be ``exposed' by applying $R^{-1}$ to the $j^{th}$ standard basis vector). Then a straightforward argument shows that a matrix $M$ with integer entries has inverse with  integer  entries iff $\det(M) = \pm 1$. 
    This completes the proof.
    
\end{enumerate}
\end{proof}

\section{Appendix: Dual version of the \perfeq }

\label{app:dual_equations}

The polynomial time algorithm for identifying perfect 3XOR games proposed in~\cite{watts20203xor} was based on a system of equations ``dual'' to the defining equations
\begin{align}
  \Gamma z  =   S \pmod{2}. 
\end{align}
Here we review the theory behind these equations. An implementation of this ``dual'' algorithm (which is less efficient in practice than \Cref{alg:QandC_perf}) is described in \Cref{sec:MmaPrograms}.

The dual equations are:
    \begin{align}
        \label{eq:3XOR_dual_SOE}
        M \xi=b \qquad M:=
        \begin{pmatrix}
            \Gamma\adj & 0 \\
            S\adj & 2
        \end{pmatrix} \quad 
        b 
        &= 
        \begin{pmatrix}
        0 \\
        1
        \end{pmatrix}.
    \end{align}
    %%%%
    \iffalse

    \emph{does not} have a solution over $\mathbb{Z}_2$.\footnote{See \Cref{app:duality} for a more complete discussion of duality.}
    
Surprisingly, the same `type' equations govern \qperf strategies.
In \cite{watts2018algorithms,watts20203xor} it was shown (using a long argument)  that a 3XOR game does not have a perfect quantum strategy iff the system of equations \Cref{eq:3XOR_dual_SOE} does not have a solution over $\mathbb{Z}$.
\fi
%%%%%
    
\noindent Note the zero appearing in the definition of $M$ has dimensions $3n$ by $1$ and the zero appearing in the definition of $b$ has dimensions $m$ by $1$. 

The  quantum and classical satisfiability conditions dual to \Cref{thm:QperfHerm_app}
are given
in the following theorem, 
and were known earlier.

\begin{theorem}[\cite{watts2018algorithms,watts20203xor}]
    \label{thm:3XOR_solns}
     \mbox{} \\
     \begin{enumerate}
     \item \label{it:diophsoln}
     Finding a perfect quantum strategy for 
     $\cG$ 
     is not possible
     iff $M \xi=b$ has a solution $\xi$ whose entries are all 
      integers.
     \item 
     \label{it:z2soln}
     Finding a perfect classical strategy for $\cG$  is not possible
     iff $M \xi =b$ has a solution  $\xi \in {\mathbb Z}_2$.
     %mod 2.
         \end{enumerate}
    \end{theorem}
We note that \Cref{thm:QperfHerm_app}
is that it produces a MERP
solution while \Cref{thm:3XOR_solns} only governs existence of solutions.

\ssec{Classical Strategy: Proof of \texorpdfstring{\Cref{thm:3XOR_solns}  \Cref{it:z2soln}}{Theorem 6 Item 2}}

Moving from the defining equations of a 3XOR 
to \Cref{it:z2soln} of \Cref{thm:3XOR_solns} is done through a standard duality argument, which we repeat here for completeness. 

\begin{proof}[Proof (\Cref{it:z2soln} of \Cref{thm:3XOR_solns}).]
\label{pf:thm1 it2}
All equations in this proof are equations over $\mathbb{Z}_2$ and so we omit the $(\!\!\!\! \mod 2)$ notation used elsewhere in this paper. 

Because the image of a matrix is the orthogonal compliment of its adjoint's kernel (i.e. $\text{Im}(M) = \text{Ker}(M\adj)^\perp$ for any finite dimensional matrix $M$) we have \Cref{eq:rep_3XOR_dual_SOE} is equivalent to the statement that for all $\xi \in \left(\mathbb{Z}_2\right)^{3m}$ 
\begin{align}
         \Gamma\adj 
    \xi = 0 \implies S\adj \xi = 0
\end{align}
This happens iff there does not exist a vector $\xi \in \left(\mathbb{Z}_2\right)^{3m}$ with
\begin{align}
         \Gamma \adj 
    \xi = 0  \ \ \ \ \  \text{   and   } \ \ \ \ \ \
    S\adj \xi = 1
\end{align}
or equivalently iff the system of equations 
\begin{align}
    \label{eq:rep_3XOR_dual_SOE}
   M \xi:=  \begin{pmatrix}
        \Gamma \adj & 0 \\
        S\adj & 2
    \end{pmatrix} 
    \xi 
    &= 
    \begin{pmatrix}
    0 \\
    1
    \end{pmatrix}
\end{align}
does not have a solution over $\mathbb{Z}_2$. (The last column of this system of equations is essentially a dummy line, both sides are identically zero over $\mathbb{Z}_2$, but we leave it in to line up the classical and quantum cases). 
\end{proof}

\subsection{Quantum Strategy: Proof of \texorpdfstring{\Cref{thm:3XOR_solns}  \Cref{it:diophsoln}}{Theorem 6 Item 1}}

\Cref{it:diophsoln} of \Cref{thm:3XOR_solns} follows from \Cref{cor:perfect_3XOR_merp} and an  duality argument (Thm 5.30 of \cite{watts2018algorithms}).

We start the proof by constructing
a dual to the system of equations in \Cref{it:diophsoln}, 
using Corollary 4.1a of \cite{schrijver1998theory}:
{\it
the linear system of equations
$G z =b$ with $G, b$ having rational entries
has a solution $z$ with integer entries iff for each rational vector~$x$: 
\begin{align*}
 x^T G \text{ has integer entries}
\implies 
 x^T b \text{ has integer entries } 
\end{align*}
}

We apply this to the system of equations
$M\xi=b$ in \Cref{eq:rep_3XOR_dual_SOE}
%
\iffalse
    \begin{align}
        \label{eq:rep_3XOR_dual_SOE}
        M \xi=b \qquad M:=
        \begin{pmatrix}
            A\adj & 0 \\
            B\adj & 0 \\
            C\adj & 0 \\
            S\adj & 2
        \end{pmatrix} \quad 
        b 
        &= 
        \begin{pmatrix}
        0 \\
        0\\
        0 \\
        1
        \end{pmatrix}
    \end{align}
    %
    \fi 
    %
with $\xi \in \mathbb{Z}$. By Corollary 4.1a of of Schjiver this system of equations does not have a solution iff there exists a rational vector $w$ with
\begin{align}
    M\adj w 
    &= 
    \begin{pmatrix}
         \Gamma & S \\
          0_{3n} & 2
    \end{pmatrix}
    w
    \in \mathbb{Z}^{m+1} \text{ and } \label{eq:Duality_condn_1}\\
    b\adj w &= \begin{pmatrix}
         0_{3n} & 1
    \end{pmatrix}w
    &\notin \mathbb{Z}. \label{eq:Duality_condn_2}
\end{align}
We see the the last row of \Cref{eq:Duality_condn_1} and \Cref{eq:Duality_condn_2} are both satisfied iff 
\begin{align}
    w = 
    \begin{pmatrix}
         w' \\ 
         z' + \frac{1}{2}
    \end{pmatrix}
\end{align}
with $w'$ arbitrary and $z' \in \mathbb{Z}$. Then, substituting this back into \Cref{eq:Duality_condn_1} we see \Cref{eq:Duality_condn_1,eq:Duality_condn_2} are satisfied iff there exists a $w' \in \mathbb{R}^{3n}$, $z' \in \mathbb{Z}$ with
\begin{align}
    & \begin{pmatrix}
         \Gamma  & S
    \end{pmatrix}
    \begin{pmatrix}
         w' \\
         z' + \frac{1}{2} 
    \end{pmatrix}
    \in \mathbb{Z}^{m+1} \\
   & \Leftrightarrow  2
   \begin{pmatrix}
         \Gamma & S
    \end{pmatrix}
    \begin{pmatrix}
         w' \\
         z' + \frac{1}{2}
    \end{pmatrix}= 0 \pmod{2} \\
   & \Leftrightarrow  
   \begin{pmatrix}
         \Gamma  
    \end{pmatrix}
    \begin{pmatrix}
         2w' 
    \end{pmatrix}= S(2z' + 1) \pmod{2} \\
  &  \Leftrightarrow  
   \begin{pmatrix}
         \Gamma  
    \end{pmatrix}
    \begin{pmatrix}
         2w' 
    \end{pmatrix}= S \pmod{2}.
\end{align}
Setting $2w' = z$ we see the resulting system of equations is exactly the same as the system of equations in \Cref{eq:merp_matrix_SOE}. Thus, \Cref{eq:merp_matrix_SOE} has a solution (and a 3XOR game has a perfect quantum strategy) iff \Cref{eq:rep_3XOR_dual_SOE} does not have a solution, as desired. 

 %%%%
 %%%%%%
 \section{APPENDIX: Description of Code used in Experiments}
 \label{sec:MmaPrograms}
 
In this section we sketch (in the context of our experiments) how our algorithms were implemented in Mathematica.

\subsection{Generating Random XOR Games}

We fix $m,n$ and randomly  generate
`sample size' admissible 
$A,B,C,S$ matrices (with 1's uniformly distributed). In this sampling process, we do not allow (we discard) repeated clauses. To accommodate 
the imperatives of computer time we typically take smaller sample size for larger $m,n$.
 
\subsection{Primal Method}

\Cref{alg:QandC_perf} was implemented in Mathematica 
in a straightforward way; here is 
a rough outline of the commands. 
Given the parameters specifying the number of players and clauses of the quantum game, we first create a random matrix $W:=(\Gamma,S)$ representing the linear system we try to solve. Then, we use the command 
% \begin{verbatim}
%     HermiteDecomposition[W] = {\Omega, H}
% \end{verbatim}
\begin{center}
{\tt HermiteDecomposition[W] = $\{\Omega, H\}$ }
\end{center}
to obtain the desired Hermite normal form $H$ of $W$
with $\Omega$ being the corresponding row operation matrix.

After this we partition $H$ into 4 parts
 \begin{align}
  \label{eq:hermDec2}
 H= \begin{pmatrix}
      R  & \bo \\
      0  & \bt  
 \end{pmatrix} 
 \end{align}
as in \Cref{eq:hermDec}. 
The game is quantum solvable iff all entries of $\bo$ are even integers.  Assuming real solvability, to
determine classical solvability, we simply solve the system 
\Cref{eq:matrix_SOE} for integers modulo 2. Specifically, 
% \begin{verbatim}
%     Quiet[Check[LinearSolve[\Gamma,S,Modulus->2];
%                 True, False, LinearSolve::nosol]]
% \end{verbatim}
\begin{center}
    {\tt Quiet[Check[LinearSolve[$\Gamma$,S,Modulus$\rightarrow 2$];
              \mbox{ \qquad  } True, False, LinearSolve::nosol]]  }
\end{center}

Note that the Mathematica command HermiteDecomposition may result in a memory leak. An effective fix is to do computations after setting
the option command:
\begin{center}
    {\tt SetSystemOptions["LinearAlgebraOptions"$\rightarrow  \{$"IntegerBlasParallelizationThresholds"$\rightarrow\{$2*d,2*d,8$\}$] }
\end{center}
where $d = \max(m,n)$ for a $m\times n$-matrix input.

\subsection{Dual Method}

Now we  give the main commands used in our implementation of the ``dual'' method for identifying \qperf 3XOR games outlined in \Cref{app:dual_equations}.

We use the Mathematica command 
\begin{verbatim}
    Solve[equation, Modulus -> 2] 
\end{verbatim}
to find binary solutions to the 
\Cref{eq:matrix_SOE}
or to show none exists. If a binary solution does not exist (there is no classical 3XOR solution), we turn to $Mx=b$ in \Cref{eq:3XOR_dual_SOE}. On this we use the Command
\begin{verbatim}
    FindInstance[equation, variables, Integers]
\end{verbatim}
to find integer solutions (meaning the game has no quantum solution)
or show this is impossible (meaning there is a quantum solution).
%  We look for patterns in the data this produces as $m$ and $n$ vary
%  and we report many of these is.
%  \\

\ssec{Our use of the Algorithms}

We compared the results coming from both the Primal and the Dual Algorithm and found that, as expected, they always agree.
The first algorithm we implemented and used was the Dual Algorithm and the data generated here for  less than 33 questions came from it.
Data coming from more than 32 questions was produced by the Primal Algorithm.

\end{document}